\documentclass[a4paper,UKenglish,cleveref, autoref, thm-restate]{lipics-v2021}

\pdfoutput=1 
\hideLIPIcs  

\title{Decomposition and Preprocessing of Ternary Constraint Networks}

\author{Pierre Talbot}{University of Luxembourg, Luxembourg}{pierre.talbot@uni.lu}{https://orcid.org/0000-0001-9202-4541}{}

\authorrunning{P. Talbot} 

\Copyright{Pierre Talbot} 


\category{} 

\relatedversion{} 

\nolinenumbers 

\usepackage{times}
\usepackage{soul}
\usepackage{url}
\usepackage{graphicx}
\usepackage{amsmath}
\usepackage{amsthm}
\usepackage{booktabs}
\usepackage{algpseudocode}

\usepackage{amssymb}
\usepackage{syntax}
\usepackage{tikz}
\usepackage{stmaryrd}
\usepackage{listings}
\usepackage{pgf}
\usepackage{multicol}

\def\lett{\mathbf{let}~}
\def\linn{~\mathbf{in}~}

\usepackage{mathtools}

\newcommand\eqdef{\triangleq}

\newcommand{\tifq}{%
  \begin{tikzpicture}[baseline={(0,0.02)}]
    \draw (0,0) circle (0.03); 
    \draw (0,0.05) -- (0,0.1); 
    \draw (0,0.1) .. controls (0.25,0.22) and (-0.05, 0.4) .. (-0.07,0.15); 
  \end{tikzpicture}%
}

\newcommand{\tifdotdot}{%
  \begin{tikzpicture}
    \draw (0,0) circle (0.03); 
    \draw (0,0.1) circle (0.03); 
  \end{tikzpicture}%
}

\newcommand{\tifbar}{%
  \begin{tikzpicture}[baseline={(0,0.07)}]
    \draw (0,0) -- (0,0.35);
    \draw (0.05,0) -- (0.05,0.35);
    \draw (0,0) -- (0.05,0);
    \draw (0,0.35) -- (0.05,0.35);
  \end{tikzpicture}%
}

\newcommand{\tif}[3]{\llparenthesis~ #1 ~\tifq #2 ~\tifdotdot~ #3 ~\rrparenthesis}
\newcommand{\tifbegin}[2]{\llparenthesis~ #1 ~\tifq #2}
\newcommand{\tifcase}[2]{\tifbar~ #1 ~\tifq #2}
\newcommand{\tifend}[1]{\tifdotdot~ #1 ~\rrparenthesis}

\begin{document}

\maketitle

\begin{abstract}
Constraint programming is a general and exact method based on constraint propagation and backtracking search.
We provide a function decomposing a constraint network into a \textit{ternary constraint network} (TCN) with a reduced number of operators.
TCNs are not new and have been used since the inception of constraint programming, notably in constraint logic programming systems.
This work aims to specify formally the decomposition function of discrete constraint network into TCN and its preprocessing.
We aim to be self-contained and descriptive enough to serve as the basis of an implementation.
Our primary usage of TCN is to obtain a regular data layout of constraints to efficiently execute propagators on graphics processing unit (GPU) hardware.
\end{abstract}

\section{Introduction}

Constraint programming is a general and exact method based on constraint propagation and backtracking search~\cite{lecoutre-constraint-2009}.
We provide a function decomposing a constraint network into a \textit{ternary constraint network} (TCN) with a reduced number of operators (Section~\ref{ternary-constraint-network}).
TCN are not new and have been used since the inception of constraint programming, notably in constraint logic programming systems~\cite{van-hentenryck-constraint-chip-1989,codognet-compiling-1996,benhamou-real-1997}.
Our contribution is to specify the decomposition function of discrete constraint network into TCN (Section~\ref{ternary-constraint-network}).
Furthermore, as the size of the decomposition is up to 3 orders of magnitude larger than the initial constraint network, we apply several known preprocessing techniques to reduce its size (Section~\ref{preprocessing}).
Over the 100 instances of the MiniZinc 2024 challenge~\cite{stuckey-minizinc-2014}, the median increase is 4.8x in the number of variables and 4.5x in the number of constraints, although 12 instances remain more than 100x larger.
Our paper aims to be self-contained and descriptive enough to serve as the basis of an implementation.
We implemented the decomposition and preprocessing of MiniZinc constraint networks in our constraint solver Turbo~\cite{AAAI2022-Talbot}.

Our primary usage of TCN is to obtain a regular data layout of constraints to efficiently execute propagators on graphics processing unit (GPU) hardware.
Another use-case of TCN is for education.
We have used this formalism to teach constraint programming and program solvers that can be used on MiniZinc instances~\cite{nethercote-minizinc:-2007}.
TCNs are useful to describe propagation and focus on the fundamental of constraint programming.

\section{Constraint Programming}
\label{cp-background}

In the following, we consider constraint programming over integer variables only.
Let $X$ be a finite set of variables and $C$ be a finite set of constraints.
For each constraint $c \in C$, let $\mathit{scp}(c) \subseteq X$ be the set of free variables of $c$, called its \textit{scope}---for instance, $\mathit{scp}(x < y) = \{x, y\}$.
Without loss of generality, we represent the domain of variables using intervals.
Let $I = \{[\ell,u] \;|\; \ell \in \mathbb{Z} \cup \{-\infty\}, u \in \mathbb{Z} \cup \{+\infty\}, \ell \leq u\} \cup \{\bot\}$ be the set of intervals ordered by inclusion with a special element $\bot$ representing the empty interval.
We define $\mathit{lb}([\ell,u]) \eqdef \ell$ and $\mathit{ub}([\ell,u]) \eqdef u$ to extract the lower and upper bounds.
A \textit{constraint network} is a pair $P = \langle d, C \rangle$ such that $d \in X \to I$ is the \textit{domain function}.
We denote $\mathbf{D}$ the set of all domain functions $X \to I$ ordered pointwise ($d \leq d' \Leftrightarrow \forall{x \in X},~d(x) \subseteq d'(x)$).
An \textit{assignment} is a map $\mathit{asn}: X \to \mathbb{Z}$, and we denote the set of all assignments by $\mathbf{Asn}$.
The set of solutions of a constraint is given by $\mathit{rel}(c) \subseteq \mathbf{Asn}$.
The set of solutions of a constraint network is:
\begin{displaymath}
\begin{array}{l}
\mathit{sol}(d, C) \eqdef \{\mathit{asn} \in \mathbf{Asn} \;|\; \forall{c \in C},~\mathit{asn} \in \mathit{rel}(c) \land \forall{x \in X},~\mathit{asn}(x) \in d(x)\}
\end{array}
\end{displaymath}

An interval propagator is a function $p_c: \mathbf{D} \to \mathbf{D}$ where $c \in C$ is a constraint.
Let $d \in \mathbf{D}$, then a propagator is reductive ($p_c(d) \leq d$), monotone ($d \leq d' \Rightarrow p_c(d) \leq p_c(d')$) and sound ($\mathit{sol}(d, \{c\}) \subseteq \mathit{sol}(p_c(d), \{\})$).
It is also complete on singleton intervals: whenever $\forall{x \in \mathit{scp}(c)},~\exists{v \in \mathbb{Z}},~d(x) = [v,v]$, then $\mathit{sol}(d, \{c\}) \supseteq \mathit{sol}(p_c(d), \{\})$.

Constraint propagation consists in finding the greatest fixpoint of a set of propagators $\{p_1,\ldots,p_n\}$ over a domain $d \in X \to I$.
As long as the propagators are executed fairly, their order of execution does not matter and the same greatest fixpoint is always eventually reached~\cite{apt-essence-1999}.
This fact has been used to design various \textit{propagation algorithms} to accelerate the computation of the fixpoint~\cite{schulte-efficient-2008,propagation-guido-tack-2009}.
As constraint propagation is sound but incomplete in general, it must be interleaved with a search procedure.
Let $\mathit{split} \in \mathbf{D} \to \mathcal{P}(\mathbf{D})$ be a strictly reductive ($\forall{d \in \mathbf{D}},~\forall{d' \in \mathit{split}(d)},~d' < d$) and sound and complete ($\forall{d \in \mathbf{D}}, \bigcup \{\mathit{sol}(d', C) \;|\; d' \in \mathit{split}(d)\} = \mathit{sol}(d, C))$) branching procedure.
We introduce next a standard optimization algorithm based on the \textit{propagate-and-search} constraint solving algorithm.
The following algorithm finds a solution $\mathit{best} \in \mathbf{D}$ which minimizes the value of the variable $z \in X$.
\begin{algorithmic}
\Function{\textnormal{$\mathit{minimize}$}}{$d, C, \mathit{best}, z$}
\State $d(\mathit{z}) \leftarrow d(\mathit{z}) \cap [-\infty, \mathit{lb(best(z))}-1]$
\State $\mathit{propagate}(d, C)$
\If{$\mathit{isasn}(d)$}
  \State $\mathit{best} \leftarrow d$
\ElsIf{$\lnot \mathit{isbot}(d)$}
  \State $\forall{d' \in \mathit{split}(d)},~\mathit{minimize}(d', C, \mathit{best}, z)$
\EndIf
\EndFunction
\end{algorithmic}
with
\begin{itemize}
  \item $\mathit{propagate}(d, \{c_1,\ldots,c_n\})$ computes the greatest fixpoint of $p_{c_1} \circ \ldots \circ p_{c_n}$ below $d$.
  \item $\mathit{isasn}(d) \eqdef \forall{x \in X},~\exists{v \in \mathbb{Z}},~d(x) = [v,v]$
  \item $\mathit{isbot}(d) \eqdef \exists{x \in X},~d(x) = \bot$
\end{itemize}
\noindent
Here and thereafter, we always pass parameters by reference.
The function $\mathit{isasn}$ tests if $d$ maps only to interval singletons (assignment) and $\mathit{isbot}$ tests whether a variable has an empty domain, in which case we must backtrack.
It is well-known that the algorithm \texttt{minimize} is a sound and complete solving procedure, see e.g.~\cite{lecoutre-constraint-2009,propagation-guido-tack-2009}.
The result holds even in the presence of infinite intervals as long as they become finite after a finite number of propagation steps.

\paragraph*{Representation of Propagators}

It is possible to construct an infinite number of constraints, take for example the sequence $\langle x_1 = x_2, x_1 = x_2 \odot_1 x_3, x_1 = x_2 \odot_1 x_3 \odot_2 x_4, \ldots \rangle$ for a sequence of arithmetic operators $\langle \odot_1, \odot_2, \ldots\rangle$.
Therefore, we cannot implement a different propagator function for each possible constraint.
Historically, solvers have limited their constraint languages in order to avoid implementing too many different propagators~\cite{van-hentenryck-constraint-chip-1989,codognet-compiling-1996,benhamou-real-1997}.
More complex constraints must be rewritten into supported constraints automatically or by the user.
The \textit{indexicals} approach is a particularly concise way of specifying propagators by mean of a few primitive constructions~\cite{carlsson-indexicals-1997,cc-fd-impl-1998}.

However, the decomposition of constraints into primitive ones increase the number of auxiliary variables and propagators.
It has been shown to be detrimental to the solver performance~\cite{schulte-view-based-2013,correia-view-based-2013}.
Therefore, the key idea is to implement a propagator as an interpreter over the abstract syntax tree (AST) of the constraint.
The propagator recursively traverses the AST to evaluate each node and compute the domain of each subexpression of the constraint.
It is called a \textit{view-based propagator} in recent work~\cite{schulte-view-based-2013,correia-view-based-2013}, but a similar technique was already used in the HC-4 consistency algorithm~\cite{hull-box-consistency-1999}.

Modern constraint solvers such as Choco and OR-Tools implement view-based propagation using inheritance to represent the AST, and subtype polymorphism to evaluate the tree.
To avoid the overhead caused by dynamic subtyping, parametric polymorphism has also been used~\cite{correia-view-based-2013,schulte-view-based-2013}, but it has the inconvenient that the solver must be recompiled for each new constraint problem tackled.

\section{Ternary Constraint Network}
\label{ternary-constraint-network}

In this section, we rewrite any constraint network into a \textit{ternary constraint network} (TCN), that is, a constraint network with only constraints of arity 3 such as $x = y + z$ and $b = (x \leq y)$ with $\{x,y,z,b\} \subseteq X$.
\begin{definition}
A ternary constraint network $\langle d, C \rangle$ is a constraint network such that each $c \in C$ is of the form $x = y \odot z$ where $x,y,z \in X$ are variables and $\odot$ is a binary operator.
\end{definition}
\noindent
In this paper, the set of supported operators is $\odot \in \mathit{OP} = \{+,*,/,\mathit{mod},\mathit{min},\mathit{max},=,\leq\}$\footnote{We follow MiniZinc semantics by using truncated integer division and Euclidean modulus.}.
The constraint language considered is sufficient to support all instances of the 2024 MiniZinc challenge.
Unary constraints of the form $x = k$, $x \leq k$ and $x \geq k$ where $x \in X$ and $k \in \mathbb{Z}$ are directly represented as domains of the variables.
The definition of the set $\mathit{OP}$ could be different.
For instance, the constraint $x = \mathit{min}(y, z)$ can be rewritten into $-x = \mathit{max}(-y, -z)$ which, when decomposed into ternary constraints, gives $\exists{x'},\exists{y'},\exists{z'},~x' = \mathit{max}(y', z') \land x' = \mathit{zero} - x \land y' = \mathit{zero} - y \land z' = \mathit{zero} - z$, with $d(\mathit{zero}) = [0,0]$.
It allows $\mathit{OP}$ to be smaller, but it introduces three new variables and constraints.
The $\mathit{min}$ and $\mathit{max}$ are particularly important since they also encode disjunction and conjunction, and it is preferable to keep them both.
The lack of subtraction is justified by the relational semantics of constraints as we can rewrite $x = y - z$ into $y = x + z$ without loss of precision.

We now give a complete decomposition of a constraint network into a TCN, along with the proof that a TCN has exactly the same set of solutions than the initial constraint network.
Note that the global constraints are decomposed in primitive constraints before this decomposition and using the MiniZinc compiler in the experiments.

Let $\langle d \in X \to I, C \rangle$ be a constraint network.
We introduce the function $\mathit{tcn}(d, C) = \langle d', C' \rangle$ rewriting a constraint network into a ternary constraint network.
The constraint language considered for $C$ is sufficient to support all instances of the 2024 MiniZinc challenge.
The function $\mathit{tcn}$ only adds new variables and therefore we have $\mathit{dom}(d) \subseteq \mathit{dom}(d')$ where $\mathit{dom}$ is the domain of the function $d$.
For clarity, we denote by $\tif{b}{X}{Y}$ the function returning $X$ if the expression $b$ is true and $Y$ otherwise.
We also use the \textit{let expression} $\lett x = E \linn E'$ which binds the result of the evaluation of the expression $E$ to the variable $x$, and returns the evaluation of the expression $E'$ where $E'$ can use $x$.
It extends naturally to the case where $E$ returns a tuple of values, e.g. $\lett x, y = E \linn E'$.

We first define several functions to extend a constraint network with a new variable $x \notin \mathit{dom}(d)$, and to update a variable already in $d$:
\begin{displaymath}
\begin{array}{l}
\mathit{extend}(d, x, [\ell, u]) = (d \cup \{x \mapsto [\ell, u]\}, x)\\
\mathit{extend}(d, x) = \mathit{extend}(d, x, [-\infty, \infty])\\
\mathit{extend}(d) = \mathit{extend}(d, \mathit{fresh}(d))\\
\mathit{extend}_\mathbb{B}(d) = \mathit{extend}(d, \mathit{fresh}(d), [0,1])\\
\mathit{extend}_\mathit{co}(d, k) = \\
  \qquad \lett c = \mathtt{\_\_CONSTANT\_}\tif{k < 0}{\mathtt{m}k}{k} \linn \\
  \qquad \tif{c \in \mathit{dom}(d)}{(d, c)}{\mathit{extend}(d, c, [k,k])}\\[0.2cm]
\mathit{update}(d, x, [\ell, u]) = \\
\quad \{y \mapsto \tif{x=y}{[\ell, u] \cap d(y)}{d(y)} \;|\; y \in \mathit{dom}(d)\}
\end{array}
\end{displaymath}

\noindent
The function $\mathit{fresh}(d)$ returns a new variable's name $x$ such that $x \notin \mathit{dom}(d)$.
We use the function $\mathit{extend}_\mathbb{B}$ to create a new Boolean variable.
To avoid creating different variables with the same constant value, we create a unique name for each constant.
For instance, \texttt{__CONSTANT_5} is the name of the variable for the constant $5$, and \texttt{__CONSTANT_m1} for the constant $-1$.
We suppose that no variable's name in the initial constraint network is of this form.
The function $\mathit{update}(d, x, [\ell, u])$ intersects the domain $[\ell_x, u_x]$ of a variable $x$ with an interval $[\ell, u]$, where interval intersection is defined as $[\ell_x, u_x] \cap [\ell, u] \eqdef [\max \{\ell_x, \ell\}, \min \{u_x, u\}]$.

The function $\mathit{tcn}(d, C) = \langle d', C' \rangle$ rewrites each constraint to be ternary:
\begin{displaymath}
\begin{array}{l}
\mathit{tcn}(d, \{\}) = \langle d, \{\} \rangle \\[0.2cm]
\mathit{tcn}(d, \{c_1, c_2, \ldots, c_n\}) = \\
\qquad \lett d, T = \mathit{tcn}(d, \{c_2, \ldots, c_n\}) \linn \\
\qquad \lett d, U, x = \mathit{tc}(d, c_1) \linn \\
\qquad \langle \mathit{update}(d, x, [1,1]), T \cup U \rangle
\end{array}
\end{displaymath}
Each constraint $c \in C$ is rewritten into a set $T$ of TCN constraints, possibly with new variables in $d$, using the function $\mathit{tc}$.
The variable $x$ is a Boolean variable reifying the constraint $c$.
As we are in the top-level conjunction, we must set this variable to $1$ to activate the constraint which is done by updating its domain.


We introduce the recursive function $\mathit{tc}(d, t) = (d', T, x)$ which rewrites a term or constraint $t$ into a set $T$ of TCN constraints.
The result of the expression $t$ (or its reification status if it is a constraint) is stored in the variable $x$.
We first consider the base cases (variables and constants) and unary constraints (arithmetic negation $-$, set membership $\in$ and absolute value function $\mathit{abs}$).
Variables are simply returned as any variable occurring in a constraint must already be in $d$.
A new variable is created for each distinct constant $k$ using $\mathit{extend}_\mathit{co}$ defined above.
Unary constraints are rewritten using equivalent ternary constraints.
For the set membership, we need the function $\mathit{itvs}(S)$ to turn a set $S$ into a set of intervals, for instance $\mathit{itvs}(\{1,2,3,5\}) = \{[1,3], [5,5]\}$.
The rewriting strategy is always the same: we first rewrite the parameters of the function or predicate, and then assemble the results to rewrite the current expression.

\begin{displaymath}
\begin{array}{l}
\mathit{tc}(d, x) = (d, \{\}, x)\\[0.2cm]
\mathit{tc}(d, -t) = \mathit{tc}(d, 0 - t) \\[0.2cm]
\mathit{tc}(d, t \in S) = \\
\qquad \lett d, T, x = \mathit{tc}(d, t) \linn \\
\qquad \lett d = \mathit{update}(d, x, [\min S, \max S]) \linn \\
\qquad \lett d, U, y = \\
\qquad\quad\mathit{tc}(d, \bigvee_{[\ell,u] \in \mathit{itvs}(S)} (x \geq \ell \land x \leq u)) \linn \\
\qquad (d, T \cup U, y) \\[0.2cm]
\mathit{tc}(d, k) = \\
\qquad \lett d, x = \mathit{extend}_\mathit{co}(d, k) \linn \\
\qquad (d, \{\}, x) \\[0.2cm]
\mathit{tc}(d, \mathit{abs}(t)) = \\
\qquad \lett d, T, x = \mathit{tc}(d, t) \linn \\
\qquad \lett d, U, y = \mathit{tc}(d, 0 - x) \linn \\
\qquad \lett d, V, z = \mathit{tc}(d, \mathit{max}(x,y)) \linn \\
\qquad \lett d = \mathit{update}(d, z, [0, \infty]) \\
\qquad \mathit{tc}(d, T \cup U \cup V, z) \\
\end{array}
\end{displaymath}

The rewriting of binary operators $\odot \in \mathit{OP}$ is factored into a single function.
When the operator is Boolean ($\leq,=$), the result is stored into a Boolean variable, otherwise in an integer variable.
The other operators can be rewritten into expressions containing operators in $\mathit{OP}$.
In particular, as subtraction is the inverse of addition, we have $x = y - z \Leftrightarrow y = x + z$.
It is not the case of multiplication and division over integers, which is why we keep both operators.
We rewrite the operator $\neq$ by negating equality.
As we could be in a reified context, we cannot simply add the TCN constraint $\mathit{zero} = (x = y)$, which is why we delegate the rewriting to logical equivalence.

\begin{displaymath}
\begin{array}{l}
\mathit{tc}(d, t_1 \odot t_2) = \\
\quad \lett d, x = \\
\qquad \tif{\odot \in \{\leq,=\}}{\mathit{extend}_\mathbb{B}(d)}{\mathit{extend(d)}} \linn \\
\quad \lett d, T, y = \mathit{tc}(d, t_1) \linn \\
\quad \lett d, U, z = \mathit{tc}(d, t_2) \linn \\
\quad (d, T \cup U \cup \{x = y \odot z\}, x) \\[0.2cm]
\mathit{tc}(d, t_1~\neq~t_2) = \\
\qquad \lett d, \mathit{zero} = \mathit{extend}_\mathit{co}(d, 0) \linn \\
\qquad \mathit{tc}(d, \mathit{zero} \Leftrightarrow t_1 = t_2)
\end{array}
\end{displaymath}
\begin{displaymath}
\begin{array}{l}
\mathit{tc}(d, t_1~-~t_2) = \\
\qquad \lett d, x = \mathit{extend}(d) \linn \\
\qquad \lett d, T, y = \mathit{tc}(d, t_1) \linn \\
\qquad \lett d, U, z = \mathit{tc}(d, t_2) \linn \\
\qquad (d, T \cup U \cup \{y = x + z\}, x)
\end{array}
\end{displaymath}
\begin{displaymath}
\begin{array}{l}
\mathit{tc}(d, t_1~\geq~t_2) = \mathit{tc}(d, t_2~\leq~t_1) \\
\mathit{tc}(d, t_1~>~t_2) = \mathit{tc}(d, t_2~\leq~t_1 - 1) \\
\mathit{tc}(d, t_1~<~t_2) = \mathit{tc}(d, t_1~\leq~t_2 - 1) \\
\end{array}
\end{displaymath}

The last step is to compile logical connectors into TCN constraints.
Logical formulas can appear in arithmetic expressions, e.g. $(x \neq y \land x \neq z) + (y \neq z) \geq 1$, and integer variables in logical formulas, e.g. $((x+w) \lor y) = z$.
In a Boolean context, we interpret the value of an interval $[\ell,u]$ to be $\mathit{true}$ iff $0 \notin [\ell,u]$, and $\mathit{false}$ iff $\ell = u = 0$.
However, if the constraint is reified, we must map $\mathit{true}$ to $1$ and not to any value.
Because we compile $\lor$ to the $\mathit{min}$ function, and $\land$ to the $\mathit{max}$ function, we must ensure the result of logical connectors lies in the interval $[0,1]$.
The function $\mathit{booleanize}$ maps the domain of an expression occurring in a Boolean context to a Boolean variable.
As an optimization, we only map $x$ to a Boolean variable if it is not already a Boolean variable, i.e. its domain is within the interval $[0,1]$.
This function is then used to rewrite all logical connectors.
Note that we cannot directly rewrite $t_1~\mathit{xor}~t_2$ as $t_1 \neq t_2$, otherwise expressions such as $[1,1]~\mathit{xor}~[2,2]$ would be $\mathit{true}$, which according to our semantics is not correct as both $[1,1]$ and $[2,2]$ should map to $\mathit{true}$.

\begin{displaymath}
\begin{array}{l}
\mathit{booleanize}(d, t) = \\
\qquad \lett d, T, x = \mathit{tc}(d, t) \linn \\
\qquad \mathbf{if}~\lnot (d(x) \leq [0,1])~\mathbf{then}\\
\qquad\quad \lett d, U, b = \mathit{tc}(d, x \neq 0) \linn \\
\qquad\quad (d, T \cup U, b) \\
\qquad \mathbf{else}~(d, T, x) \\[0.2cm]
\mathit{tc}(d, \lnot t) = \mathit{tc}(d, t = 0) \\[0.2cm]
\mathit{tc}(d, t_1 \Rightarrow t_2) = \mathit{tc}(d, \lnot t_1 \lor t_2) \\[0.2cm]
\mathit{tc}(d, t_1 \land t_2) = \\
\qquad \lett d, b = \mathit{extend}_\mathbb{B}(d) \linn \\
\qquad \lett d, T, b_1 = \mathit{booleanize}(d, t_1) \linn \\
\qquad \lett d, U, b_2 = \mathit{booleanize}(d, t_2) \linn \\
\qquad (d, T \cup U \cup \{b = \mathit{min}(b_1,b_2)\}, b) \\[0.2cm]
\mathit{tc}(d, t_1 \lor t_2) = \\
\qquad \lett d, b = \mathit{extend}_\mathbb{B}(d) \linn \\
\qquad \lett d, T, b_1 = \mathit{booleanize}(d, t_1) \linn \\
\qquad \lett d, U, b_2 = \mathit{booleanize}(d, t_2) \linn \\
\qquad (d, T \cup U \cup \{b = \mathit{max}(b_1,b_2)\}, b) \\[0.2cm]
\mathit{tc}(d, t_1 \Leftrightarrow t_2) = \\
\qquad \lett d, b = \mathit{extend}_\mathbb{B}(d) \linn \\
\qquad \lett d, T, b_1 = \mathit{booleanize}(d, t_1) \linn \\
\qquad \lett d, U, b_2 = \mathit{booleanize}(d, t_2) \linn \\
\qquad (d, T \cup U \cup \{b = (b_1 = b_2)\}, b) \\[0.2cm]
\mathit{tc}(d, t_1~\mathit{xor}~t_2) = \\
\qquad \lett d, T, b_1 = \mathit{booleanize}(d, t_1) \linn \\
\qquad \lett d, U, b_2 = \mathit{booleanize}(d, t_2) \linn \\
\qquad \lett d, V, b = \mathit{tc}(d, b_1 \neq b_2) \linn \\
\qquad (d, T \cup U \cup V, b) \\[0.2cm]
\end{array}
\end{displaymath}

We now prove that every constraint network $\langle d, C \rangle$ is equivalent to its rewriting $\mathit{tcn}(d, C)$, i.e., it has exactly the same set of solutions.
Because we introduce new variables, we must restrict the solutions to those defined on the initial variables.
Let $\mathit{asn}: X \to \mathbb{Z}$ be an assignment and $Y$ be a subset of $X$, its restriction $\mathit{asn}|_Y$ is defined as $\{y \mapsto \mathit{asn}(y) \;|\; y \in Y\}$.
We extend the restriction to set of assignments $A$ as follows: $A|_Y = \{\mathit{asn}|_Y \;|\; \mathit{asn} \in A\}$.
We prove the following statements by structural induction on the formula, showing that at each step the set of solutions is preserved.

\begin{proposition}[Soundness and completeness]
Let $\langle d, C \rangle$ be a constraint network defined over the set of variables $X$, then $\mathit{sol}(d, C) = \mathit{sol}(\mathit{tcn}(d, C))|_X$.
\end{proposition}

\begin{proof}
The function $\mathit{tcn}$ applies recursively $\mathit{tc}$ on each constraint.
Therefore, we first prove the soundness and completeness of $\mathit{tc}$.
Let $d_1, T, x = \mathit{tc}(d, t)$ and $Y = X \cup \{x\}$, the induction hypothesis is:
\begin{displaymath}
\mathit{sol}(d_1, T)|_Y = \mathit{sol}(d_2, \{x = t\}) \text{ with } d_2 = d \cup \{x \mapsto [-\infty, \infty]\}
\end{displaymath}
If $t$ is a top-level constraint, then $\mathit{tcn}$ sets the domain of $x$ to $1$, thus we fall back on the equality $\mathit{sol}(d_1, T \cup \{x = 1\})|_Y = \mathit{sol}(d, \{c\})$ which holds since for any constraint $c$, we have $c \Leftrightarrow (x \Leftrightarrow c \land x = 1)$.
Further, because $\mathit{sol}(d, \{c_1\}) \cap \mathit{sol}(d, \{c_2\}) = \mathit{sol}(d, \{c_1, c_2\})$, the equality holds for any constraint network $\langle d, C \rangle$.

We must prove the induction hypothesis holds for the function $\mathit{tc}$.
When appropriate, we prove the induction hypothesis by showing a logical equivalence $x = t \Leftrightarrow x = t'$, in which case we must have, by the induction hypothesis on $t$ and $t'$, the equality $\mathit{sol}(d_1, T)|_Y = \mathit{sol}(d_2, T')|_Y$ where $d_1, T, x = \mathit{tc}(d, t)$ and $d_2, T', x = \mathit{tc}(d, t')$.
\begin{itemize}
  \item Variable: the domain $d_1$ is unchanged, hence $d_1 = d_2$. The set of constraints is empty and $x = x$ is a tautology. Therefore, $\mathit{sol}(d_1, T)|_Y = \mathit{sol}(d_1, \{\})|_Y = \mathit{sol}(d_2, \{x = x\})|_Y$.
  \item Constant: by definition we have $x \in \mathit{dom}(d_2)$ and $d_1(x) = [k,k]$.
  For any constraint network $\langle d, C \rangle$ with $x \in \mathit{dom}(d)$, we have the equality $\mathit{sol}(\mathit{update}(d, x, [k,k]), C) = \mathit{sol}(d, C \cup \{x = k\})$.
  \item Negation: by logical equivalence $x = -t \Leftrightarrow x = 0 - t$.
  \item Membership: by logical equivalence of $t \in S \Leftrightarrow \bigvee_{v\in S} t = v$ and $t = v \lor t = v+1 \lor \ldots \lor t = v+n \Leftrightarrow t \geq v \land t \leq v+n$.
  Furthermore, the update of $d(x)$ with $[\min S, \max S]$ is implied by the membership constraint and therefore does not modify the set of solutions.
  \item Absolute value function: by logical equivalence $x = \mathit{abs}(t) \Leftrightarrow x = \mathit{max}(t, -t)$. Furthermore, the update of $d(x)$ with $[0, \infty]$ does not modify the set of solutions as it is implied by the tautology $\mathit{max}(t, -t) \geq 0$.
  \item TCN ternary constraint: by logical equivalence $x = t_1 \odot t_2 \Leftrightarrow x = y \odot z \land y = t_1 \land z = t_2$.
  \item Subtraction: by logical equivalence $x = t_1 - t_2 \Leftrightarrow t_1 = x + t_2$.
  \item Inequalities: by logical equivalences $x = t_1 \geq t_2 \Leftrightarrow x = t_2 \leq t_1$, $x = t_1 > t_2 \Leftrightarrow x = t_2 \leq t_1 - 1$, $x = t_1 < t_2 \Leftrightarrow x = t_1 \leq t_2 - 1$, and $x = t_1 \neq t_2 \Leftrightarrow x = (0 \Leftrightarrow (t_1 = t_2))$.
  \item Logical connectors: we first notice that logical operators are functions $\square \in \mathbb{B}^n \to \mathbb{B}$ where $n$ is the arity of the logical operator $\square$.
  The conversion from an integer value to a Boolean value is done implicitly.
  However, this conversion must be encoded explicitly when decomposing the constraints because we use arithmetic operators $\odot \in \mathbb{Z}^n \to \mathbb{Z}$ to encode the logical operators, which does not have the same domains than their logical equivalent $\square$.
  We restrict the parameter to Boolean value by converting them in $\mathit{booleanize}$.
  The conversion of $y$ to $b_1$ is given by $b_1 = 0 \Leftrightarrow y = 0$ and $b_1 = 1 \Leftrightarrow y \neq 0$, and similarly $z$ to $b_2$.
  Then we have the following equivalences:
  \begin{itemize}
    \item Logical negation: $x = \lnot t_1 \Leftrightarrow (x = (b_1 = 0))$.
    \item Conjunction: $(x = t_1 \land t_2) \Leftrightarrow (x = \max(b_1, b_2))$, and $\forall{a,b \in \mathbb{B}},~\max(a,b) \in \mathbb{B}$.
    \item Disjunction: $(x = t_1 \lor t_2) \Leftrightarrow (x = \min(b_1, b_2))$, and $\forall{a,b \in \mathbb{B}},~\min(a,b) \in \mathbb{B}$.
    \item Implication: $(x = t_1 \Rightarrow t_2) \Leftrightarrow (x = \lnot t_1 \lor t_2)$.
    \item Equivalence: $(x = t_1 \Leftrightarrow t_2) \Leftrightarrow (x = (b_1 = b_2))$.
    \item Exclusive disjunction: $(x = t_1~\mathit{xor}~t_2) \Leftrightarrow (x = (b_1 \neq b_2))$.
  \end{itemize}
\end{itemize}
\end{proof}

\begin{proposition}[Uniqueness]
Further, we have a bijection between both set of solutions: $|\mathit{sol}(d, C)| = |\mathit{sol}(\mathit{tcn}(d, C))|$.
\end{proposition}

\begin{proof}
Let $s_1,s_2$ be two solutions in $\mathit{sol}(\mathit{tcn}(d, C))$ and $s$ be a solution in $\mathit{sol}(d, C)$.
Suppose that $s_1|_X = s$ and $s_2|_X = s$.
It means that $s_1$ and $s_2$ differ on the value of at least an auxiliary variable $x$, or are the same.
We show by induction that the variables in $\mathit{dom}(s)$ fully define the auxiliary variables, and therefore $s_1$ and $s_2$ must be equal.
Our induction hypothesis is that the variable returned by $\mathit{tc}$ is fully defined.
\begin{itemize}
  \item Variable: no auxiliary variable is created.
  \item Constant: the auxiliary variable is created with a single value.
  \item Ternary constraint $x = y \odot z$: by induction hypothesis, $y$ and $z$ are fully defined.
  Since all $\odot$ are functions, it is necessary that the auxiliary variable $x$ is fully defined whenever $y$ and $z$ are.
  Note that division is fully defined over the intervals as dividing by zero maps to the empty interval.
  \item Subtraction $y = x + z$: similarly, $x$ must be fully defined whenever $y$ and $z$ are since the inverse of subtraction is a function (addition).
  \item Logical operators $b = b_1\,\square\,b_2$: new Boolean variables are only created for the conjunction, disjunction and equivalence. By induction hypothesis, $b_1$ and $b_2$ are fully defined. The operators $\min, \max$ and $=$ are functions, hence $b$ must be fully defined as well.
\end{itemize}
\end{proof}

\section{Preprocessing of Ternary Constraint Network}
\label{preprocessing}

We follow standard preprocessing techniques that we specialize to ternary constraint networks~\cite{rendl-phd-2010,nightingale-automatically-2017}. 
The goal of preprocessing is essentially to remove variables and constraints.
Before presenting our preprocessing algorithm, we define a structure to keep track of equivalent variables.

We aim at finding a partition $E$ of $X$ such that each component $Y$ of $E$ represents a set of equivalent variables.
More precisely, all pairs of variables $x, y \in Y$ are connected by an equality constraint $x = y$.
We write $[x]_E \in E$ the equivalence class of $x$ in $E$.
We suppose variables are totally ordered (e.g. by an indexing) and we choose $\min Y$ to be the representative variable of the equivalence class $Y$.

The equivalence classes are discovered by various preprocessing techniques.
Initially, we suppose the variables are all distinct which is given by the partition $\mathit{init}(X) \eqdef \{\{x\} \;|\; x \in X\}$.
We add a variable equality $x = y$ by removing both equivalence classes $[x]_E$ and $[y]_E$ from $E$ and adding back their union into the partition.
\begin{displaymath}
\begin{array}{l}
\mathit{merge}(E, x, y) \eqdef \\
 \lett XY = \{[x]_E\} \cup \{[y]_E\} \linn
 (E \setminus XY) \cup \{\bigcup XY\}
\end{array}
\end{displaymath}
The interval domain of a variable $x \in X$ in an equivalence class $Y$ is the intersection of all variables' domains in $Y$, defined as $d_E(x) \eqdef \bigcap_{y \in [x]_E} d(y)$.
During preprocessing, we read the domain of a variable using $d_E$ instead of $d$.
In practice, we implement the partition efficiently using a union-find data structure.

We apply seven preprocessing functions:
\begin{itemize}
  \item \textit{Propagation} to reduce the domains of the variables.
  \item \textit{Algebraic simplification} to eliminate constraints for which the solutions can be expressed directly in the domains of the variables. It also detects equivalence between variables and refine the partition $E$.
  \item \textit{Common subexpression elimination} detects ternary constraints with the same right-hand side. For instance, consider $a = y + z$ and $b = y + z$, the procedure eliminates one of the two constraints and merge the equivalence classes of $a$ and $b$.
  \item \textit{Merge domains} of the variables in the equivalence classes. It is especially useful for propagation which is not aware of the equivalence classes.
  \item \textit{Entailed constraint elimination} to remove the constraints that are entailed by $d$.
  For instance, if $d(x) = [1,2]$ and $d(y) = [2,3]$, then the ternary constraint $1 = (x \leq y)$ is always true regardless of the evolution of $d$.
  \item \textit{Variable renaming} to use a unique variable per equivalence class.
  \item \textit{Useless variable elimination} to remove variables not in the scope of any constraint. We keep the variables with an empty domain to be able to detect unsatisfiability.
\end{itemize}

The preprocessing steps are combined in the computation of a greatest fixpoint over the triple $\langle d, C, E \rangle$.
We stop once $d$ and $E$ do not change anymore, and we ignore the modifications on $C$.
To be able to define a greatest fixpoint, we must define a partial order on partitions.
Let $E, E'$ be two partitions such that $\bigcup E = \bigcup E'$, then we define the following partial order: $E \leq E' \Leftrightarrow \forall{Y \in E},~\exists{Z \in E'},~Y \supseteq Z$.
The intuition is that we only merge equivalence classes and never divide an existing one while preprocessing.
We ignore the change on $C$ because, by inspecting the places where we rewrite constraints, the rewritten constraint cannot trigger further rewriting without a change in $d$ or $E$.
Therefore, if $d$ and $E$ do not change, all constraints that can be rewritten must have been rewritten already.

The detection of entailed constraints and useless variables, as well as variable renaming, are extracted outside of the fixpoint loop.
Indeed, regardless of the modifications on $d$ and $C$, an entailed constraint will stay entailed, and a variable not occurring in the scope of any constraint will remain useless.
Formally, preprocessing is defined by the following two algorithms.

\begin{displaymath}
\begin{array}{l}
\mathit{preprocess}(d, C) \eqdef \\
\quad \lett E = \mathit{init}(\mathit{dom}(d)) \linn \\
\qquad\qquad \qquad\qquad \qquad\qquad \hfill \text{(\textit{initialize equivalence classes})}\\
\quad \lett \langle d, C, E \rangle = \mathbf{gfp}_{\langle d, C, E\rangle}~\mathit{preprocess}\\
\hfill \text{(\textit{see overloaded def. below})} \\
\quad \lett C = C \setminus \{c \in C \;|\; \gamma(d) \subseteq \mathit{rel}(c)\} \linn \\
\hfill \text{(\textit{eliminate entailed constraints})} \\
\quad \lett R = \{x \mapsto \min~[x]_E \;|\; x \in \mathit{dom}(X)\} \linn \\
\hfill \text{(\textit{create a substitution})}\\
\quad \lett C = \{c[R] \;|\; c \in C\} \linn \\
\hfill \text{(\textit{rename variables by their representative elements})}\\
\quad \lett d = \{x \mapsto d(x) \;|\; x \in \bigcup_{c \in C} \mathit{scp}(c) \lor d(x) = \bot\} \linn \\
\hfill \text{(\textit{eliminate useless variables})} \\
\quad \langle d, C \rangle
\end{array}
\end{displaymath}

\begin{displaymath}
\begin{array}{l}
\mathit{preprocess}(d, C, E) \eqdef \\
\quad \lett d = \mathbf{gfp}_d~p_{c_1} \circ \ldots \circ p_{c_n} \linn \\
 \qquad\qquad\qquad\qquad\qquad\qquad\qquad\quad \hfill \text{(\textit{root propagation})}\\
\quad \lett \langle d, C, E \rangle = \mathit{as}(d, C, E) \linn \\
\hfill \text{(\textit{algebraic simplification})} \\
\quad \lett \langle d, C, E \rangle = \mathbf{gfp}_{\langle d, C, E \rangle}~\mathit{icse} \linn \\
\hfill \text{(\textit{common subexpression elimination})} \\
\quad \lett d = \{x \mapsto \mathit{dom}(E, d, x) \;|\; x \in X\} \linn \\
\hfill \text{(\textit{merge domains of equivalent variables})} \\
\quad (\langle d, C \rangle, E)
\end{array}
\end{displaymath}
\noindent
Note that in practice, we must save the partition $E$ and eliminated variables in order to print the solutions with the original variables of the model, but this poses no particular challenge.

We give an example of the algorithm on the constraint network $\langle d, \{x = y + z, w = y + z, x = (y = z)\} \rangle$ with $d(x) = [0,1]$, $d(w) = [1,2]$ and all other variables with domain $[-\infty, \infty]$.
Propagation is unable to remove any value from the domains and $\mathit{as}$ does not detect any equality constraint since $x$ could be equal to $0$.
Then, $\mathit{icse}$ detects the equality $x = w$, modifies the partition to $\{\{x,w\}, \{y\}, \{z\}\}$ and removes the constraint $x = (y = z)$.
The propagation is still inefficient, but this time $\mathit{as}$ is able to detect the equality $y = z$ since $d_E(x) = d(x) \cap d(w) = [1,1]$, and the partition is updated to $\{\{x,w\}, \{y,z\}\}$.
At the next iteration, $\mathit{as}$ rewrites the constraint $w = y + z$ into $w = y * 2$.
It enables the propagation step to reduce the domain of $y$ to $\bot$, effectively detecting the problem unsatisfiable.
At that point, the constraint network is $\langle d, \{w = y * 2\} \rangle$ and the fixpoint of $\mathit{preprocess}$ is reached.
The entailment checking step removes the constraint $w = y * 2$ as $\gamma(d) = \{\}$.
In this situation, the problem has been completely solved, and detected unsatisfiable, by preprocessing.


\paragraph*{Common Subexpression Elimination}

We eliminate the TCN subexpressions that are identical by iterating over each pair of constraints in $C$.

\begin{displaymath}
\begin{array}{l}
\mathit{icse}(d, \{\}, E) = \langle d, \{\}, E \rangle \\
\hfill \text{(\textit{base cases $|C| \leq 1$})}\\
\mathit{icse}(d, \{c\}, E) = \langle d, \{c\}, E \rangle\\[0.2cm]

\mathit{icse}(d, \{x_1 = (y_1 \odot_1 z_1), c_2, \ldots, c_n\}, E) = \\
\quad \lett \langle d, C, E \rangle = \mathit{icse}(d, \{c_2, \ldots, c_n\}, E) \linn \\
\quad \mathbf{if}~\exists{(x_i = (y_i \odot_i z_i)) \in C},~\odot_1 = \odot_i \\
\quad \land ((y_1 = y_i \land z_1 = z_i) \lor (\odot_1 \in \mathbf{C} \land y_1 = z_i \land z_1 = y_i))\\
\quad \mathbf{then} \\
\hfill \text{(\textit{detecting subexpression equality})}\\
\qquad \langle d, C, \mathit{merge}(E, x_1, x_i) \rangle \\
\hfill \text{(\textit{$c_1$ is removed and an equality is added})}\\
\quad \mathbf{else}\\
\qquad \langle d, C \cup \{x_1 = (y_1 \odot_1 z_1)\}, E \rangle \\
\end{array}
\end{displaymath}
\noindent
Note that we implement this algorithm efficiently in $\mathcal{O}(|C|)$ by using a hash map between $y \odot z$ and $x$ for each constraint $x = (y \odot z)$.
To account for commutativity, the hashing function must give the same result for $y \odot z$ and $z \odot y$, which can be done by simply multiplying the indexes of the variables and the operator.
The equality function of the hash map must compare the elements taking into account commutativity as shown in the algorithm above.

\paragraph*{Algebraic Simplification}

We present a few rewriting rules that are not taken into account when performing interval propagation, especially because propagation does not deal well with multiple occurrences of the same variable in a constraint.
The function $\mathit{as}$ is essentially a pattern matching algorithm over the constraints.
We write $\mathbf{C} =\{+,*,\min,\max,=\}$ the set of commutative operators.
To simplify the notation, we rewrite each constraint in a normal form as follows:
\begin{displaymath}
\begin{array}{l}
\mathit{nf}(E, d, x = y \odot z) = \\
\quad \lett y, z = \tif{\odot \in \mathbf{C} \land y \leq z}{(y,z)}{(z,y)} \linn \\
\hfill \text{(\textit{normalize commutative operators})} \\
\quad \lett y, z = \tif{\odot \in \mathbf{C} \land d_E(y) = [k,k]}{(z,y)}{(y,z)} \linn \\
\hfill \text{(\textit{constant on the right of }}\odot\text{)} \\
\quad \lett x, y, z = \min~[x]_E, \min~[y]_E, \min~[z]_E \linn \\
\hfill \text{(\textit{variable substitution})}\\
\quad (x = y \odot z)
\end{array}
\end{displaymath}
Furthermore, an integer $\mathbf{k}$ in bold font denotes a variable $x$ such that $d(x) = [k,k]$.

\begin{displaymath}
\begin{array}{l}
\mathit{as}(d, \{\}, E) = \langle d, \{\}, E \rangle \\[0.2cm]
\mathit{as}(d, \{c_1, c_2, \ldots, c_n\}, E) = \\
\quad \lett d, \mathbf{2} = \mathit{extend}_{\mathit{co}}(d, 2) \linn \\
\quad \lett d, C, E = \mathit{as}(d, \{c_2, \ldots, c_n\}) \linn \\
\quad \mathbf{match}~\mathit{nf}(c_1)~\mathbf{with}\\
\quad |~x = x + y \to \langle \mathit{update}(d,y,[0,0]), C, E \rangle \\
\quad |~x = y + \mathbf{0} \to \langle d, C, \mathit{merge}(E, x, y) \rangle \\
\quad |~x = y + y \to \langle d, C \cup \{x = y * \mathbf{2}\}, E \rangle \\
\quad |~x = x * \mathbf{k} \to \\
\qquad\quad\tif{k = 1}{\langle d, C, E \rangle}{\langle \mathit{update}(d,x,[0,0]), C, E \rangle} \\
\quad |~\mathbf{k} = x * x \to \\
\qquad\quad \tif{\exists{n\in\mathbb{N},~n*n = k}\\ \qquad\quad }
  {\langle \mathit{update}(d,x,[-\sqrt{k},\sqrt{k}]), C \cup \{c_1\}, E \rangle\\ \qquad\quad}
  {\langle \mathit{update}(d,x,\bot), C, E \rangle}\\
\quad |~x = y * \mathbf{1} \to \langle d, C, \mathit{merge}(E, x, y) \rangle  \\
\quad |~x = x * x \to \langle \mathit{update}(d,x,[0,1]), C, E \rangle \\
\quad |~x = \mathbf{1} / x \to \langle \mathit{update}(d,x,[-1,1]), C \cup \{c_1\}, E \rangle \\
\quad |~x = \mathbf{0} / x \to \langle \mathit{update}(d,x,\bot), C, E \rangle \\
\quad |~\mathbf{k} = x / x \to \\ \qquad \quad
  \tif{k = 1 \\ \qquad \quad }
  {\langle d, C \cup \{c_1\}, E \rangle\\ \qquad\quad }
  {\langle \mathit{update}(d,x,\bot), C, E \rangle} \\
\quad |~x = y / \mathbf{1} \to \langle d, C, \mathit{merge}(E, x, y) \rangle \\
\quad |~x = x / x \to \langle \mathit{update}(d,x,[1,1]), C, E \rangle \\
\quad |~x = x~\mathit{mod}~x \to \langle \mathit{update}(d,x,[0,0]), C, E \rangle \\ 
\quad |~x = x~\mathit{mod}~\mathbf{k} \to \\
\qquad\quad \langle \mathit{update}(d,x,[0, \mathit{abs}(k)-1]), C, E \rangle \\
\quad |~x = \mathbf{k}~\mathit{mod}~x \to \langle \mathit{update}(d,x,\bot), C, E \rangle \\
\quad |~\mathbf{0} = x~\mathit{mod}~x \to \langle d, C, E \rangle \\
\quad |~x = \min(y, y) \to \langle d, C, \mathit{merge}(E, x, y) \rangle \\
\quad |~x = \max(y, y) \to \langle d, C, \mathit{merge}(E, x, y) \rangle \\
\quad |~x = \min(x, y) \to \langle d, C \cup \{\mathbf{1} = (x \leq y)\}, E \rangle \\
\quad |~x = \max(x, y) \to \langle d, C \cup \{\mathbf{1} = (y \leq x)\}, E \rangle \\

\quad |~\mathbf{1} = (x = y) \to \langle d, C, \mathit{merge}(E, x, y) \rangle \\
\quad |~x = (y = y) \to \langle \mathit{update}(d,x,[1,1]), C, E \rangle \\
\quad |~x = (x = \mathbf{k}) \to \langle \tifbegin{k = 0}{\mathit{update}(d,x,\bot)}\\
\qquad\qquad\qquad\qquad\quad\;\;\; \tifcase{k = 1}{\mathit{update}(d,x,[1,1])}\\
\qquad\qquad\qquad\qquad\quad\;\;\; \tifend{\mathit{update}(d,x,[0,0])}, C, E \rangle\\
\qquad |~x = (y \leq y) \to \langle \mathit{update}(d,x,[1,1]), C, E \rangle \\
\qquad |~x = (x \leq \mathbf{k}) \to \langle \tifbegin{k = 0}{\mathit{update}(d,x,\bot)} \\
\qquad\qquad\qquad\qquad\quad\;\;\; \tifcase{k > 0}{\mathit{update}(d,x,[1,1])} \\
\qquad\qquad\qquad\qquad\quad\;\;\; \tifend{\mathit{update}(d,x,[0,0])}, C, E \rangle \\

\qquad |~x = (\mathbf{k} \leq x) \to \langle \tifbegin{k = 1}{d} \\
\qquad\qquad\qquad\qquad\quad\;\;\; \tifcase{k < 1}{\mathit{update}(d,x,[1,1])} \\
\qquad\qquad\qquad\qquad\quad\;\;\; \tifend{\mathit{update}(d,x,[0,0])}, C, E \rangle \\

\end{array}
\end{displaymath}

There are a few rules we left out because they are already implicitly encoded by constraint propagation over intervals.
It is the case of the absorbing element $0$ for multiplication and division.
Note that there is no need to evaluate constants as this is already taken into account by propagation and the removal of entailed constraints.

\section{Analysis of Ternary Constraint Networks}

\begin{table}
\begin{tabular}{lllllllll}
\toprule
  \multicolumn{1}{l}{} & \multicolumn{4}{c}{\textbf{Variables}} & \multicolumn{4}{c}{\textbf{Constraints}} \\
\multicolumn{1}{l}{} & \textbf{average} & \textbf{median} & \textbf{stddev} & \textbf{max} & \textbf{average} & \textbf{median} & \textbf{stddev} & \textbf{max} \\
\midrule
FlatZinc & 9.42x & 1.86x & 18.63x & 111.62x & 24.94x & 2.95x & 67.91x & 486.87x \\
TCN & 53.61x & 7.97x & 151.61x & 1133.22x & 76.68x & 6.21x & 265.88x & 1837.19x \\
Preprocessed & 22.06x & 4.76x & 50.48x & 344.62x & 36.39x & 4.33x & 115.18x & 746.09x \\
\bottomrule
\end{tabular}
\caption{Increase in size of constraint networks relatively to Choco constraint networks over 89 instances.}
\label{size-tcn}
\end{table}

\begin{figure*}
    \centering
    \includegraphics[scale=0.6]{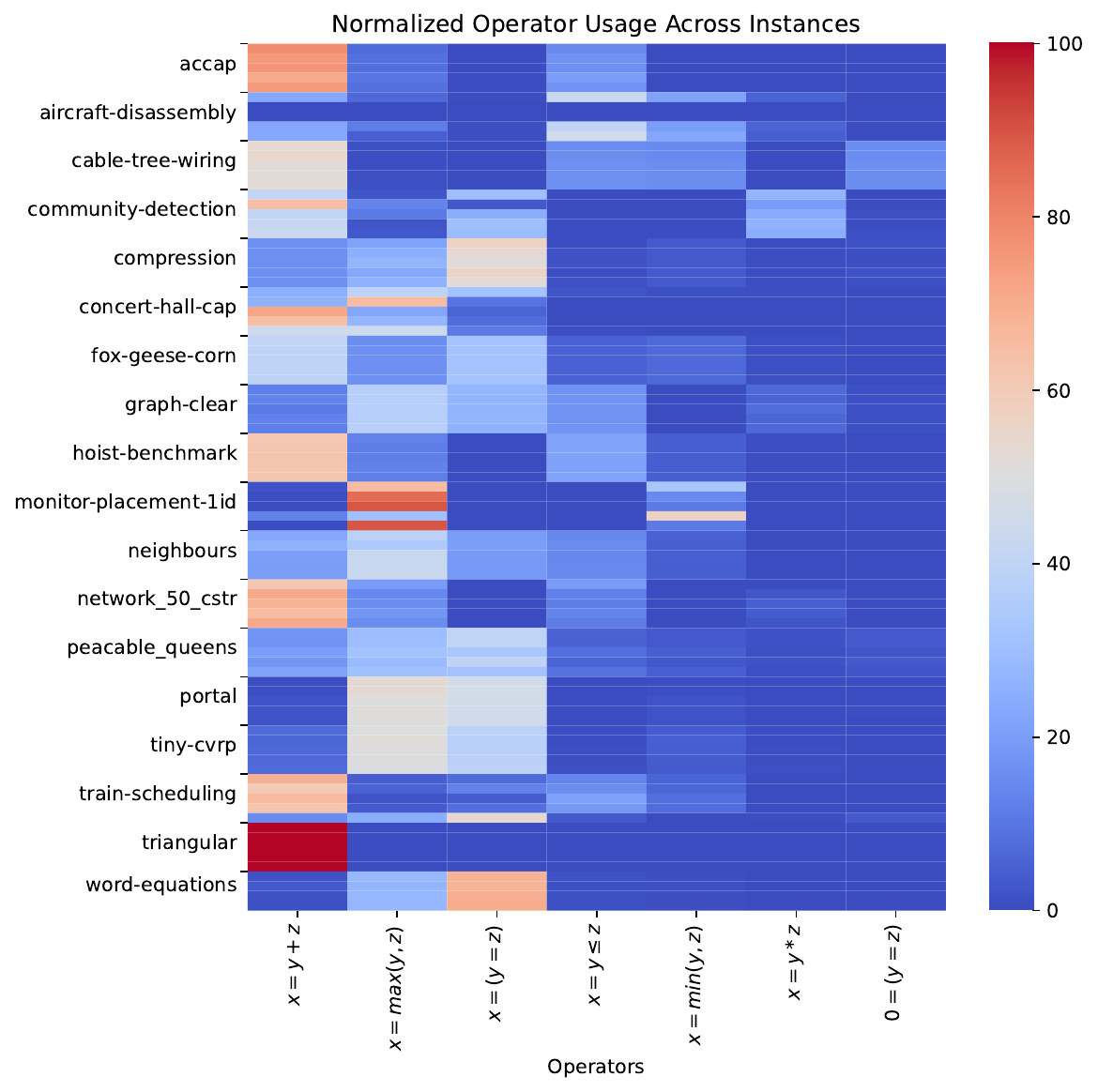}
    \caption{Heatmap of the operators used in preprocessed ternary constraint networks over 89 instances grouped by problem classes.}
    \label{ops-heatmap}
\end{figure*}

\begin{figure*}
  \centering
  \includegraphics[scale=0.7]{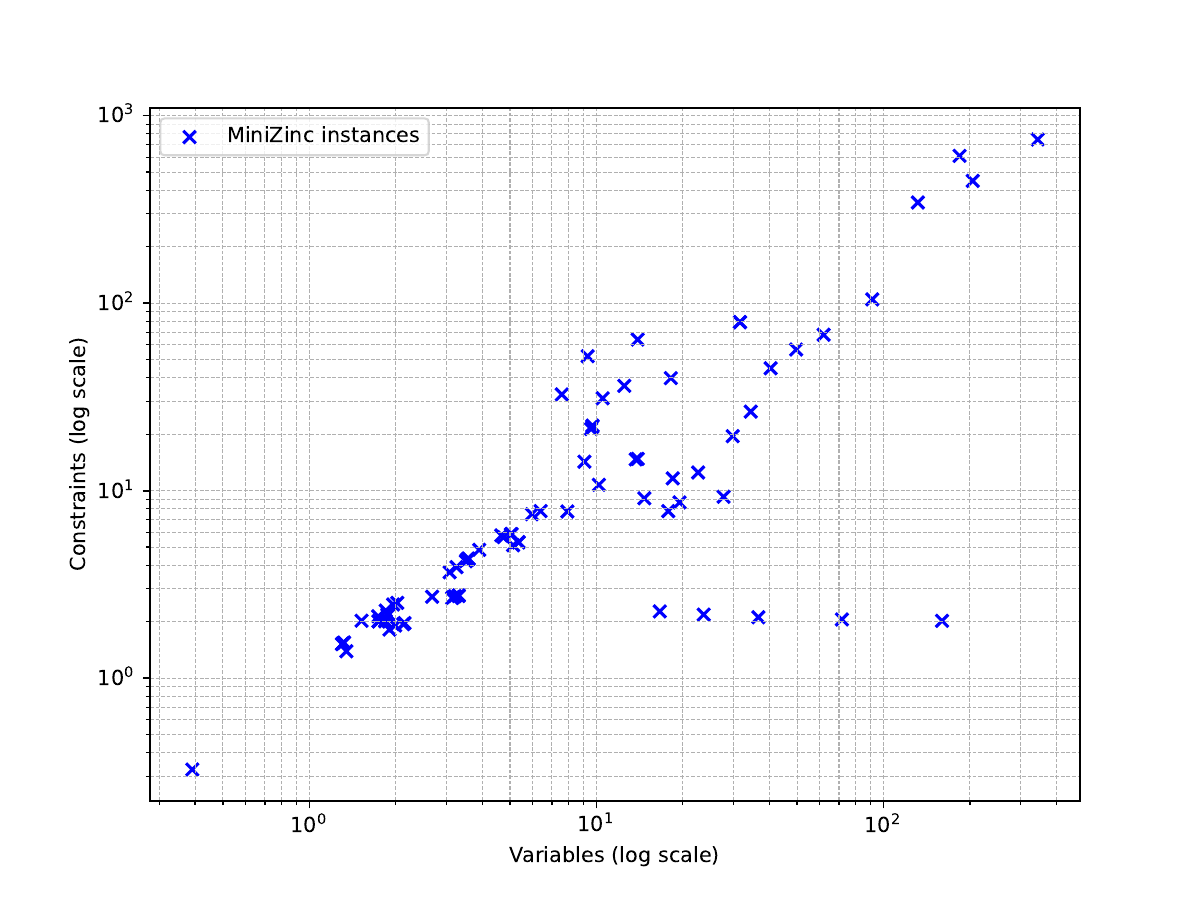}
  \caption{Increase in size of preprocessed ternary constraint networks relatively to Choco constraint networks (95 instances).}
  \label{tnf-increase-scatter-plot}
\end{figure*}

As we do not support any global constraint, the MiniZinc to FlatZinc conversion is already quite costly.
Table~\ref{size-tcn} shows the average, median, standard deviation and maximum increase in the numbers of variables and constraints for fully decomposed MiniZinc model (without global constraints), after decomposition into TCN (Section~\ref{ternary-constraint-network}) and after preprocessing (Section~\ref{preprocessing}).
On 63/89 instances, the FlatZinc decomposition is less than an order of magnitude larger than the Choco constraint network in both the number of variables and constraints.
After applying the TCN decomposition, we further increase by 5 times the number of variables and 3 times the number of constraints on average.
Fortunately, the preprocessing step reduces this increase to 2.5 times for variables and 1.5 times for constraints on average, which is a reasonable increase considering the constraint network only contains ternary constraints.
For a more precise overview, we provide a scatter plot of the variables and constraints increases for all instances in Figure~\ref{tnf-increase-scatter-plot}.
The average preprocessing time is 24.22s with a standard deviation of 96.91s.
The median time is 0.91s and only 11 instances take more than 10 seconds to be preprocessed.

On Figure~\ref{ops-heatmap}, we analyze the usage of the operators across instances.
Clearly, thread divergence is not be the main concern anymore since the instances are not using a wide variety of operators.
In particular, there is no instance with modulo and division operations (or they are simplified after preprocessing).
Because of the preprocessing, all non-reified equality constraints $1 = (x = y)$ are deleted.
The constraint $x = y \leq z$ only exists in a reified context and $0 = y \leq z$ (modelling $y > z$) and $1 = y \leq z$ have completely disappeared.
It is a consequence of the FlatZinc and TCN decomposition, $y \leq z$ is rewritten $y - z \leq 0$ by the FlatZinc decomposition, which is further rewritten $y = x + z \land x \leq 0$ leading to $\leq$ and $>$ being directly managed in the domain of the variables.
The most used TNF constraints are addition, maximum which is used to encode disjunction, and reified equality.
Interestingly, although present in 5 problem classes, the decomposition of all-different into $\mathcal{O}(n^2)$ constraints of the form $0 = (x = y)$ does not seem to be a bottleneck.



\bibliography{tcn}

\end{document}